\documentclass[11pt]{article}
\usepackage{color, graphicx, amssymb, amsmath, amsthm, fullpage, dsfont}
\usepackage{tikz,pgf,braids}
\usetikzlibrary{backgrounds}
\usepackage{color}
\usepackage[colorlinks=true,urlcolor=webblue,linkcolor=webgreen,filecolor=webblue,citecolor=webgreen,pdfpagemode=UseOutlines,pdfstartview=FitH,pdfpagelayout=OneColumn,bookmarks=true]{hyperref}

\definecolor{webgreen}{rgb}{0,.5,0}
\definecolor{webblue}{rgb}{0,0,.5}

\input{Qcircuit}

\title{Classical simulation of Yang-Baxter gates}
\author{Gorjan Alagic\footnote{Institute for Quantum Information and
    Matter, California Institute of Technology, Pasadena, CA.}
~and Aniruddha Bapat\footnote{California Institute of Technology,
    Pasadena, CA.}
~and Stephen Jordan\footnote{National Institute of
    Standards and Technology, Gaithersburg,
    MD. \texttt{stephen.jordan@nist.gov}}}
\date{}

\def\identity{I}
\def\CC{\mathbb C}

\def\EE{\mathbb E}

\def\opnm{\operatorname}

\newcommand{\re}{\text{Re}}
\newcommand{\im}{\text{Im}}

\newtheorem{theorem}{Theorem}

\newtheorem{definition}{Definition}
\newtheorem{lemma}{Lemma}

\begin{document}
\bibliographystyle{plain}
\maketitle

\begin{abstract}
A unitary operator that satisfies the constant Yang-Baxter equation immediately yields a unitary representation of the braid group $B_n$ for every $n \geq 2$. If we view such an operator as a quantum-computational gate, then topological braiding corresponds to a quantum circuit. A basic question is when such a representation affords universal quantum computation. In this work, we show how to classically simulate these circuits when the gate in question belongs to certain families of solutions to the Yang-Baxter equation. These include all of the qubit (i.e., $d = 2$) solutions, and some simple families that include solutions for arbitrary $d \geq 2$. Our main tool is a probabilistic classical algorithm for efficient simulation of a more general class of quantum circuits. This algorithm may be of use outside the present setting.
\end{abstract}

\section{Introduction}

The Yang-Baxter equation, named after C. N. Yang and R. J. Baxter, appears in a number of areas of mathematics and physics. Yang encountered the equation while working on two-dimensional quantum field theory, while Baxter applied it to exactly solvable models in statistical mechanics~\cite{Baez92}. An accessible review of some of the many applications of the Yang-Baxter equation can be found in~\cite{PerkYang06}. In this work, we will consider what is typically called the constant quantum Yang-Baxter equation, and is defined as follows. Let $V$ be a finite-dimensional complex Hilbert space and $R$ a linear operator on $V \otimes V$. Then $R$ satisfies the quantum Yang-Baxter equation (YBE) if
$$
(R \otimes \identity) (\identity \otimes R)(R \otimes \identity) = (\identity \otimes R)(R \otimes \identity)(\identity \otimes R)\,,
$$
where $\identity$ denotes the identity operator on $V$. In this case, we say that $R$ is a Yang-Baxter operator. The YBE bears a close resemblance to the relation 
$$
\sigma_i \sigma_{i+1} \sigma_i = \sigma_{i+1} \sigma_i \sigma_{i+1}
$$
of the braid group $B_n$. Indeed, a Yang-Baxter operator naturally gives the space $V^{\otimes n}$ the structure of a representation $\rho_{(R, n)}$ of $B_n$. Turaev showed that if $R$ also satisfies the so-called Markov property, then it corresponds to an invariant of links~\cite{Turaev88}. The invariant is given by the (appropriately scaled) trace of $\rho_{(R, n)}$, evaluated at any braid whose trace closure is equal to the link. More generally, one can derive a link invariant from the trace of any representation of $B_n$ which satisfies the Markov property. This is the case for the famous Jones representation and the corresponding Jones Polynomial invariant~\cite{Jones85}.
Freedman, Kitaev, Larsen and Wang~\cite{FKW02, FLW02a, FLW02b} showed that the Jones representation has significant meaning in quantum computation. Informally speaking, the Jones representation provides a functionality-preserving ``dictionary'' between quantum circuits and braids. One consequence of these results is that additively approximating the Jones Polynomial is a universal problem for quantum computation. It also appears that this dictionary could correspond to a physically plausible implementation of quantum computers by means of exotic particles called non-abelian anyons~\cite{Preskill04}. One downside of the Jones representation in this context is that topological locality of braiding does not translate naturally into tensor-product locality of the corresponding quantum circuit. In particular, it is not the case that braiding two adjacent strands correponds to applying a Yang-Baxter operator on the space of two adjacent qubits. One might hope that the Jones representation could be made to look this way, e.g., by changing bases or manipulating the multiplicities of its irreducible summands. However, Rowell and Wang recently showed that this is impossible unless the Jones representation in question\footnote{Recall that, just like the Jones polynomial, the Jones representation has a parameter (in addition to $n$) which is typically a root of unity. Quantum universality holds for most but not all values of this parameter.} is in fact \emph{not} quantum-universal (see Corollary 4.2 in~\cite{RowellWang12}.)

Alternatively, one may ask if there exist other representations of the braid groups with the desired local structure and which exhibit computational universality. This amounts to finding unitary solutions to the YBE and determining if they are universal gates. In this work, we investigate low-dimensional solutions with this motivation in mind. All of the qubit (i.e., $d = \dim V=2$) solutions to the YBE were found by Hietarinta~\cite {Hietarinta93}; the unitary ones among those were identified by Dye~\cite{Dye03}. It was previously known that, when their eigenvalues are roots of unity, these solutions yield braid group representations with finite image~\cite{FRW06, Franko10}. We show how to classically approximate the matrix entries of any quantum circuit constructed from a particular kind of two-qudit gate. Most of the qubit solutions to the YBE, as well as some solution families of arbitrary dimension, are special cases of this gate. For the remaining qubit solutions, we give a different result: how to classically simulate a quantum computation that begins in any product state, and ends with a measurement of an observable on logarithmically many qubits. This is typically considered sufficient to rule out quantum universality. However, some caution is called for: there are gate sets which are known to be classically simulable in this sense but become hard to simulate when one is allowed to measure all the output qubits in the computational basis~\cite{JozsaVanDenNest,bremner}.

We remark that, as pointed out by Lomonaco and Kauffman~\cite{KauffmanLomonaco04}, some qubit solutions to the YBE are entangling gates, and any entangling gate together with arbitrary single-qubit gates is universal~\cite{Brylinski02}. However, in that case we are no longer computing with representations of the braid group. Indeed, a primary motivation for the topological approach to quantum computation is to rely on the topological stability of braiding for fault-tolerance. Applying single-qubit gates fault-tolerantly as part of this approach would require additional ideas. For this reason, we restrict ourselves to just one gate, which acts on two qubits and is a solution to the YBE. Some classes of entangling gates that have previously been shown to be classically simulatable are given in~\cite{miyake, Gottesman}.

\section{Preliminaries}

\subsection{Gates, circuits, and universality}

We briefly review basic notions about quantum gates, circuits, and computational universality. For more details, we refer the reader to the text of Nielsen and Chuang~\cite{NielsenChuang00}. Given an integer $d \geq 1$, let $[d] = \{0, 1, \dots, d-1\}$. Let $V = \CC[d]$ be a $d$-dimensional complex Hilbert space with distinguished orthonormal basis $\{\ket i :i \in [d]\}$. We refer to copies of $[d]$ as dits and copies of $V$ as qudits. For any $k$ and any $x \in [d]^k$, set $\ket{x} = \ket{x_1} \otimes \ket{x_2} \otimes \cdots \otimes \ket{x_k}$. The space $V^{\otimes k}$ has a preferred basis $\{\ket{x} : x \in [d]^k\}$, which we will call the computational basis. A unitary operator on $V^{\otimes k}$ is called a $k$-qudit gate. 

Let $\mathcal R$ be a set of gates which act on $k$ or fewer qudits. Fix $n > 0$ and, for each $l$-qudit gate $R \in \mathcal R$, define $R_j \in \opnm U(V^{\otimes n})$ to be the operator that applies $R$ to qudits $j, \ldots, j+l$ and the identity operator $\identity$ to the rest. Define $\mathcal R^{(n)}$ to be the set of all $R_j$, for every $R \in \mathcal R$ and every valid index $j$. An $n$-qudit quantum circuit over the gate set $\mathcal R$ (or $\mathcal R$-circuit for short) is a finite sequence
$$
C = (U_1,  U_2, \ldots, U_m)
$$
where for each $i$, $U_i \in \mathcal R^{(n)}$ or $U_i^{-1} \in \mathcal R^{(n)}$. We will sometimes denote the number of gates in the circuit $C$ by $|C| = m$. The circuit defines an operator
$$
C = U_m \cdot U_{m-1} \cdots U_1 \in U(V^{\otimes n})\,.
$$
Note that we have overloaded notation so that $C$ refers to both the sequence of gates and the operator implemented by their composition. Pictorially, an $\mathcal R$-circuit is represented by a diagram like the following, where each wire corresponds to one qudit. 
$$
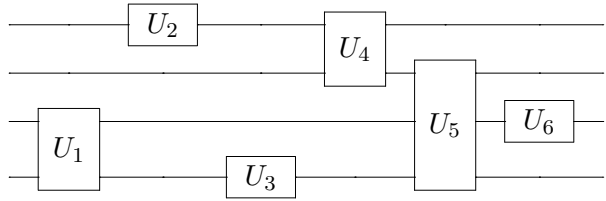
\Qcircuit @C=1em @R=.5em {
&\qw				&\gate{~U_2~}	&\qw			&\multigate{1}{U_4}	&\qw 			&\qw			&\qw\\
&\qw				&\qw 		&\qw			&\ghost{U_4}		&\multigate{2}{U_5} 	&\qw 		&\qw\\
&\multigate{1}{U_1}	&\qw 		&\qw			&\qw				&\ghost{U_5} 		&\gate{~U_6~}	&\qw\\
&\ghost{U_1}		&\qw			&\gate{~U_3~}	&\qw				&\ghost{U_5}		&\qw			&\qw
}
$$
For pictorial convenience, the gates shown in the figure only act on nearest neighbors. While the nearest-neighbor condition is needed for certain other types of circuits to be classically simulatable (e.g. matchgates~\cite{miyake}), our results do not require it. We adopt here the common convention that circuits are applied from left to right (unfortunately, the opposite of the case for operators.)
Of general interest are gate sets which allow for universal quantum computation. 

\begin{definition}
A gate set $\mathcal R$ is \textbf{universal} if there exists $N>0$ such that $N$-qudit $\mathcal R$-circuits form a dense subset of $U(V^{\otimes N})$. 
\end{definition}

\noindent The Solovay-Kitaev theorem~\cite{NielsenChuang00} tells us that, for universal $\mathcal R$, any unitary operator in $U(V^{\otimes N})$ can be approximated to precision $\epsilon$ with an $N$-qudit $\mathcal R$-circuit of length polylog$(1/\epsilon)$. Standard arguments also show that density can be extended from $N$ to any $n \geq N$.

Quantum-computational power can also be defined in terms of complexity classes. The class that is typically associated with efficient quantum computation is called BQP, which stands for bounded-error quantum polynomial time. A drawback of BQP is the lack of known complete problems, i.e., problems which are both in BQP and at least as hard (under classical polynomial-time reduction) as any other problem in BQP.  The classical analogue BPP (bounded-error probabilistic polynomial time) suffers from the same drawback. For this reason, we will work with promise versions of these two classes, i.e., PromiseBQP and PromiseBPP. We will not need the formal definitions of these classes (see, e.g.,~\cite{JanzingWocjan07}). For us it will suffice to refer to the following.

\begin{definition}
Given a set $\mathcal R$ of quantum gates, the problem $\mathcal I(\mathcal R)$ is defined as follows. Given an $n$-qudit $\mathcal R$-circuit $C$ and ditstrings $x$ and $y$, as well as a promise that either $\bra{x} C \ket{y} > 2/3$ or $\bra{x} C \ket{y}  < 1/3$, decide which is the case.
\end{definition}

\noindent We may define PromiseBQP as the class of problems which reduce to $\mathcal I(\mathcal R)$ for some universal set of quantum gates $\mathcal R$. Interestingly, there are gate sets $\mathcal R$ which are not universal in the density sense but for which $\mathcal I(\mathcal R)$ is nonetheless PromiseBQP-hard; an example is $\mathcal R = \{\text{Hadamard}, \text{Toffoli}\}$. This gate set is dense over the special orthogonal group, but since the matrix entries are all real, it cannot be dense over the unitary group.

Later on, we will show that when $\mathcal R$ consists of a single gate which belongs to certain solution families of the Yang-Baxter equation, then $\mathcal I(\mathcal R) \in $ PromiseBPP. This means that $\mathcal R$ is not quantum universal under either of the above definitions, unless the widely believed conjecture that quantum computation is more powerful than classical computation is false.

\subsection{Pauli group and Clifford group}

Recall that the single-qubit Pauli operators are defined by
$$
I =
\begin{pmatrix} 
1 & 0 \\
0 & 1 \\
\end{pmatrix}\,,
\qquad
X =
\begin{pmatrix} 
0 & 1 \\
1 & 0 \\
\end{pmatrix}\,,
\qquad
Y =
\begin{pmatrix} 
0 & -i \\
i & 0 \\
\end{pmatrix}\,,
\qquad
Z = 
\begin{pmatrix} 
1 & 0 \\
0 & -1 \\
\end{pmatrix}\,.
$$
Each Pauli operator is self-adjoint and unitary. In the $n$-qubit case, we set
$$
X_j = I^{\otimes j-1} \otimes X \otimes I^{\otimes n-j}
$$
and likewise for $Y_j$ and $Z_j$. We define the $n$-qubit Pauli group $\mathcal P_n$ to be the group generated by $\{X_j, Y_j, Z_j : j = 1, \dots, n \}$. An important property for us is that $\mathcal P_n$ spans the space of $n$-qubit Hermitian operators.

The Clifford group on $n$ qubits is defined to be the normalizer of the Pauli group inside the unitary group, i.e.,
$$
\mathcal C_n = \{ U \in U(2^n) : U P U^\dagger \in \mathcal P_n \text{ for all } P \in \mathcal P_n\}\,.
$$
By direct computation, it's easy to check that the following gates are elements of $\mathcal C_n$ for any $n \geq 2$:
$$
H = \frac{1}{\sqrt{2}}
\begin{pmatrix} 
1 & 1 \\
1 & -1 \\
\end{pmatrix}\,,
\qquad
P = 
\begin{pmatrix} 
1 & 0 \\
0 & i \\
\end{pmatrix}\,,
\qquad
CNOT = 
\begin{pmatrix} 
1 & 0 & 0 & 0 \\
0 & 1 & 0 & 0 \\
0 & 0 & 0 & 1 \\
0 & 0 & 1 & 0
\end{pmatrix}\,.
$$
It is a theorem (see \cite{Gottesman}) that the above gates, when applied to arbitrary qubits or pairs of qubits, actually generate $\mathcal C_n$. We will thus call any circuit made up of these gates a Clifford circuit. Since $\mathcal P_n \subset \mathcal C_n$, we can also add the Pauli operators to this gate set for free. We remark that the conjugation action of a Clifford circuit on an element of $\mathcal P_n$ is easy to compute in a direct, gate-by-gate fashion. For details, see~\cite{Gottesman}.

Due to the frequent appearance of $\mathcal C_n$ in various areas of quantum information, the computational power of Clifford circuits is well-studied. While $\mathcal C_n$ is finite and not universal, adding any gate outside $\mathcal C_n$ results in a universal set~\cite{Rains}. A thorough analysis of the computational power of Clifford circuits under various models is performed in~\cite{JozsaVanDenNest}.

\subsection{Yang-Baxter operators and representations of the braid group}

Let $V = \CC[d]$ and $R \in \opnm{U}(V \otimes V)$. Then $R$ satisfies the quantum Yang-Baxter equation (YBE) if
\begin{equation}\label{eq:QYBE}
(R \otimes \identity) (\identity \otimes R)(R \otimes \identity) = (\identity \otimes R)(R \otimes \identity)(\identity \otimes R)~,
\end{equation}
where $\identity$ denotes the identity operator on $V$. In this case, we say that $R$ is a Yang-Baxter operator. Let $T: \ket{a \otimes b} \mapsto \ket{b \otimes a}$ denote the swap operator on $V \otimes V$. By comparing circuit diagrams, it's not hard to see that $R$ is a solution to \eqref{eq:QYBE} if and only if $S = RT$ is a solution to
\begin{equation}\label{eq:AYBE}
S_{12} S_{13} S_{23} = S_{23} S_{13} S_{12}\,,
\end{equation}
where
$$
S_{12} = S \otimes \identity\,,
\qquad
S_{13} = (\identity \otimes T) (S \otimes \identity) (\identity \otimes T)\,,
\qquad
S_{23} = \identity \otimes S\,.
$$
Equation \eqref{eq:AYBE} is sometimes called the algebraic Yang-Baxter equation.

Recall that the braid group $B_n$ is a finitely generated group with generators $\sigma_1, \sigma_2, \cdots,\sigma_{n-1}$ and relations
$$
\begin{array}{rcll}
\sigma_i \sigma_j & = & \sigma_j \sigma_i & \forall \ |i-j| \geq 2\\
\sigma_i \sigma_{i+1} \sigma_i & = & \sigma_{i+1} \sigma_i \sigma_{i+1}
& \forall \ i.
\end{array}
$$
In 1925 Artin proved that the abstract group defined above precisely captures the topological equivalence of braided strings~\cite{Artin25}. Pictorially, braids are represented with a diagram; an example diagram for $\sigma_3^{-1}\sigma_2^{-1}\sigma_3\sigma_1^{-1}$ is shown below. We read such diagrams left-to-right, keeping the same convention as with circuits. The second generating relation of $B_n$ is known as the Yang-Baxter relation. A solution $R \in \opnm U(V \otimes V)$ of the Yang-Baxter equation yields a unitary representation $\rho_{(R, n)}$ of $B_n$ on the space $V^{\otimes n}$ for every $n$. It is defined by
$$
\rho_{(R, n)}(\sigma_i) = \identity^{\otimes (i-1)} \otimes R \otimes \identity^{\otimes{n-i-1}}\,.
$$
The images of braids under $\rho_{(R, n)}$ are precisely the $R$-circuits on $n$ qudits, where $d = \dim V$. For example, the braid $\sigma_3^{-1}\sigma_2^{-1}\sigma_3\sigma_1^{-1}$ and the corresponding $R$-circuit are shown below.
$$
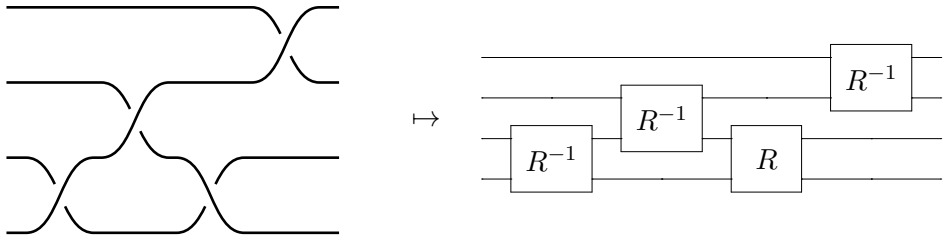
\begin{tabular}{ l c r }
      \begin{tikzpicture}
      \braid[rotate=90, line width=1pt, number of strands=4, border height=6pt] (braid) at (0,0) a_1 a_2 a_1^{-1} a_3;
      \node[at=(braid-1-s),anchor=north, left=4pt] {};
      \node[at=(braid-2-s),anchor=north, left=4pt] {};
      \node[at=(braid-3-s),anchor=north, left=4pt] {};
      \node[at=(braid-4-s),anchor=north, left=4pt] {};
      \node[at=(braid-1-e),anchor=north, left=-16pt] {};
      \node[at=(braid-2-e),anchor=north, left=-16pt] {};
      \node[at=(braid-3-e),anchor=north, left=-16pt] {};
      \node[at=(braid-4-e),anchor=north, left=-16pt] {};
    \end{tikzpicture} 
& 
\raisebox{44pt}{$\mapsto$}
& 
\raisebox{70pt}{
\Qcircuit @C=1em @R=.5em {
& \qw	& \qw	& \qw	& \multigate{1}{R^{-1}}	& \qw\\
& \qw	& \multigate{1}{R^{-1}} 	& \qw	& \ghost{R^{-1}}	& \qw\\
& \multigate{1}{R^{-1}}          & \ghost{R^{-1}} 		& \multigate{1}{~R~}	& \qw	& \qw\\
& \ghost{R^{-1}}		& \qw	& \ghost{~R~}		& \qw	& \qw
}}
\end{tabular}
$$
Under a plausible physical interpretation, a computation is performed by braiding particle-like excitations whose exchange statistics are described by $R$. If $R$ is a universal gate, this model would result in universal topological computation. Such a model could provide a basis for a quantum computer architecture with inherent fault-tolerance~\cite{Preskill04}.

\section{Classical simulation of certain quantum circuits}

In this section, we prove a general result about simulating certain quantum circuits with a classical probabilistic algorithm. We begin with two straightforward lemmas about classical sampling. (See \ref{sec:app} for proofs).

\begin{lemma}
Let $\{P_j\}_{j=1}^n$ be probability distributions on $[d]$ and let $P = \Pi_j P_j$ be the corresponding product distribution over $[d]^n$. Suppose that we can calculate $P_j(k)$ for every $j$ and every $k$ in total time $\emph{poly(n, d)}$. Then there's a classical probabilistic algorithm that runs in time $\emph{poly(n, d)}$ and samples from $[d]^n$ according to a probability distribution $D$ such that $|P - D| \leq 1/2^{\emph{poly}(n)}$.
\end{lemma}

\noindent We will also require the following Chernoff-Hoeffding bound for complex-valued random variables.

\begin{lemma}
Let $X_1, X_2, \ldots, X_n$ be independent complex-valued random variables with $\EE[X_j] = \mu$ and $|X_j| \leq b$ for all $j$. Let $S = \sum_j X_j/n$. Then
$$
\emph{Pr}\left[ \left|S - \mu \right| \geq \epsilon \right] \leq 4 \exp\left(-n\epsilon^2/8b^2\right)\,.
$$
\end{lemma}

Let $\mathcal S_d$ denote the symmetric group, i.e., the group of permutations of $d$ letters. We denote the action of $\pi \in \mathcal S_d$ on an integer $1 \leq j \leq d$ by $\pi j$. 

\begin{definition} Let $Q$ be an invertible $d \times d$ matrix over $\CC$, and $G$ a subgroup of $\mathcal S_d$. Define matrices $A, B$ by setting $A_{ij} = |Q_{ij}|$ and $B_{ij} = |(Q^{-1})_{ij}|$. We say that $Q$ satisfies property $(G)$ if for every $\pi \in G$ and every $k, l$, we have $\sum_j A_{k, \pi j}B_{jl} \leq 1$.
\end{definition}
\noindent If $Q$ is unitary, then by Cauchy-Schwarz and the orthonormality of the rows of $Q$,
$$
\sum_j A_{k, \pi j}B_{jl} \leq \left(\sum_j |A_{k, \pi j}|^2 \sum_i |B_{il}|^2\right)^{1/2} = 1.
$$
It follows that unitary matrices satisfy property $(\mathcal S_d)$. 

We are now ready to present the main classical simulation algorithm. When we refer to the matrix entries of operators in $\opnm{GL}(\CC[d]) \cong \opnm{GL}_d(\CC)$, it will always be in the computational basis. We say that such an operator is computable if its entries can be computed exactly by a classical algorithm in poly$(d)$ time. Recall that $T:a \otimes b \mapsto b \otimes a$ is the swap operator, and that for a subset $S$ of a group $G$, $\langle S \rangle$ denotes the subgroup of $G$ generated by $S$.

\begin{theorem}\label{thm:main}
Let $\mathcal R = \{R_1, R_2, \ldots, R_k\}$ be a set of unitary 2-qu$d$it gates, each one a composition
\begin{equation}\label{eq:general-gate-form}
R_i = (Q \otimes Q) D_i P_i (C_i \otimes C_i) (Q \otimes Q)^{-1}
\end{equation}
of computable, invertible operators. Suppose that for each $i$, $D_i$ is a diagonal unitary, $C_i$ is a $d \times d$ permutation matrix, and $P_i = \identity$ or $P_i = T$. Finally, let $Q$ satisfy property $(G)$ where $G = \langle \{C_i\}_{i=1}^k \rangle \leq \mathcal S_d$. Then there exists a classical probabilistic algorithm which, given an $n$-qudit $\mathcal R$-circuit $U$ and strings $x, z \in [d]^n$ and $\epsilon > 0$, outputs a number $r$ in time $\emph{poly}(n, |U|, 1/\epsilon)$ such that $|r - \langle x | U | z \rangle| < \epsilon$ except with probability exponentially small in $n$ and $1/\epsilon$.
\end{theorem}

\begin{proof}
Set $S_i = D_iP_i (C_i \otimes C_i)$. If we expand each $R_i$-gate to turn $U$ into a circuit made from $S_i$-gates and $Q$-gates, then all of the $Q$-gates except the initial and final ones are cancelled, as in the example below. We are thus left with a circuit of the form $Q^{\otimes n} V (Q^{-1})^{\otimes n}$ where $V$ is an $\{S_i\}$-circuit. We remark that, in this expanded form, the entire circuit is not necessarily a proper quantum circuit, since $Q$ might not be unitary. The circuit $V$, on the other hand, is quantum since all of its gates are unitary.
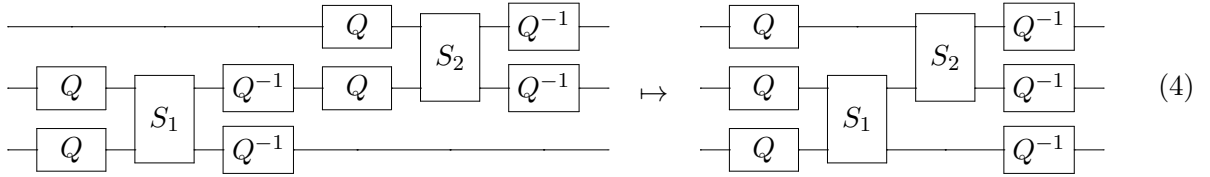
\begin{equation}\label{eq:pull-q-out}
\begin{tabular}{ l c r }
\Qcircuit @C=1em @R=.5em {
&\qw		&\qw 			&\qw			&\gate{\,~Q\,~}	&\multigate{1}{S_2} 	&\gate{Q^{-1}}	&\qw\\
&\gate{\,~Q\,~} &\multigate{1}{S_1}	&\gate{Q^{-1}}	&\gate{\,~Q\,~}	&\ghost{S_2} 		&\gate{Q^{-1}} &\qw\\
&\gate{\,~Q\,~}	&\ghost{S_1}		&\gate{Q^{-1}}	&\qw		&\qw 			&\qw 		&\qw
}& \raisebox{-28pt}{$\mapsto$} & 
\Qcircuit @C=1em @R=.5em {
&\gate{\,~Q\,~}	&\qw 			&\multigate{1}{S_2} 	&\gate{Q^{-1}}	&\qw\\
&\gate{\,~Q\,~} &\multigate{1}{S_1}	&\ghost{S_2} 		&\gate{Q^{-1}} &\qw\\
&\gate{\,~Q\,~}	&\ghost{S_1}		&\qw 			&\gate{Q^{-1}} &\qw
}
\end{tabular}
\end{equation}
Before we proceed, note that a non-nearest-neighbor gate can be written as a nearest-neighbor gate conjugated with a swap gate. We depict our gates as acting on nearest neighbors for convenience only, but this condition is not needed for the result to hold. The action of an $S_i$-gate on the $j$-th and $(j+1)$-st qudits of a computational basis state is simple to compute. The values of the two qudits are both in $[d]$ initially, and remain in $[d]$ after the action of $C_i$. Second, these new values are either swapped or left unchanged by $P_i$. Third, the $D_i$-gate adds an overall phase factor to the state. By composing these easily-computable actions, the action of $V$ on a computational basis state can be computed in time polynomial in $n$, $d$, and $|V|$. Up to phases, this action consists of permuting the $n$ qudits by some $\pi \in S_n$, and applying some bijection $f_j:[d] \rightarrow [d]$ to the initial value of the $\pi(j)$-th qudit. Each $f_j$ is a composition of $C_i$-gates, in the order specified by $V$. Explicitly, for a basis state $\ket y = \ket{ y_1 y_2 \ldots y_n}$, we write
$$
  V \ket y = e^{i\phi(y)}\ket{f_1 y_{\pi 1} \otimes f_2 y_{\pi 2} \otimes \cdots \otimes f_n y_{\pi n}}\,,
$$
where $\phi(y)$ is the overall phase resulting from the $D_i$-gates. For simplicity of notation, we denoted the image of $k$ under the permutation $\pi$ as $\pi k$, and wrote $f_j y_{\pi j}$ in place of $f_j(y_{\pi j})$. 

Next we consider the matrix element
\begin{align*}
  \bra x U \ket z = \bra{x}(Q)^{\otimes n}V(Q^{-1})^{\otimes n}\ket{z}
 	&= \sum_{y\in [d]^n} \bra{x} (Q)^{\otimes n} V \ket{y}\bra{y} (Q^{-1})^{\otimes n} \ket{z}\\ 
	&=\sum_{y\in [d]^n}e^{i\phi(y)}\prod_{j=1}^n \bra{x_j} Q \ket{f_j y_{\pi j}} \bra{y_j}Q^{-1}\ket{z_j}\,.
\end{align*}
We expand the matrix elements of $Q$ and $Q^{-1}$ in terms of magnitudes and phases:
\begin{align*}
  \bra{r}Q\ket{s} &= A(r,s)e^{i\alpha(r,s)}\\
  \bra{r}Q^{-1}\ket{s} &= B(r,s)e^{i\beta(r,s)}
\end{align*}
where $A,B,\alpha,\beta$ are real-valued and $r,s\in [d]$. Then
\begin{align*}
  \bra{x}(Q)^{\otimes n}V(Q^{-1})^{\otimes n}\ket{z} 
	&= \sum_{y\in [d]^n}e^{i\theta(y)}\prod_{j=1}^n A(x_j, f_j y_{\pi j}) B(y_j, z_j)\\
	&= \sum_{y\in [d]^n}e^{i\theta(y)}\prod_{j=1}^n A(x_{\sigma j}, f_{\sigma j} y_j) B(y_j, z_j)\,,
\end{align*}
where $\sigma = \pi^{-1}$ and 
$$
\theta(y) = \phi(y)+\displaystyle\sum_{j=1}^n\bigl(\alpha(x_{\sigma j}, f_{\sigma j}y_j) + \beta(y_j, z_j)\bigr)\,.
$$

Now we introduce the following normalization factor:
$$
\rho = \sum_{y \in [d]^n} \prod_{j=1}^n A(x_{\sigma j}, f_{\sigma j} y_j) B(y_j, z_j)
	= \prod_{j=1}^n \sum_{k \in [d]} A(x_{\sigma j}, f_{\sigma j} k) B(k, z_j)\,.
$$
This allows us to define a natural probability distribution over $[d]^n$ by
$$
P(y) = \frac{1}{\rho}\prod_{j=1}^n A(x_{\sigma j}, f_{\sigma j} y_j) B(y_j, z_j)\,,
$$
which factorizes as $P(y) = \prod_{j=1}^n P_j(y_j)$, where
$$
P_j(l) = \frac{A(x_{\sigma j}, f_{\sigma j} l) B(l, z_j)}{\sum_{k \in [d]} A(x_{\sigma j}, f_{\sigma j} k) B(k, z_j)}\,.
$$
Note that $\rho$ and all of the $P_j(l)$ can be computed in time linear in $n$ and $d$. By Lemma \ref{lem:sampling}, we can efficiently sample from $[d]^n$ according to $P$, with error exponentially small in $n$. 

In order to estimate $\bra{x}U\ket{z}$, sample repeatedly from this distribution, obtaining outcomes $\xi(j) \in [d]^n$ for $j \in \{1, 2, \ldots\}$ and output the average of the random variables $X_j := \rho \exp(i\theta(\xi(j)))$. Observe that, for each $j$,
$$
\EE[X_j] = \sum_{z \in [d]^n} \rho e^{i \theta(z)} P(z) = \bra{x}U \ket{z}\,.
$$
To control the absolute value, recall that $f_{\sigma j}$ is a composition of the permutation matrices $C_i$, and is thus an element of $\langle \{C_i\}_{i=1}^k\rangle \leq \mathcal S_d$. Since $Q$ satisfies property $(\langle \{C_i\}_{i=1}^k\rangle)$, we have
$$
|X_j|^2 = |\rho|^2 
	= \prod_{j=1}^n \Bigl| \sum_{k \in [d]} A(x_{\sigma j}, f_{\sigma j} k) B(k, z_j) \Bigr|^2
	\leq \prod_{j=1}^n 1^2 \leq 1.
$$
by Cauchy-Schwarz, for each $j$. Now set $S(r) = \sum_{j=1}^r X_j / r$. By Lemma \ref{lem:chernoff}, for $r \geq 8n/\epsilon^3$ we have
$$
\text{Pr}\left[\left|S(r) - \bra{x} U \ket{z} \right| \geq \epsilon \right] \leq 4 \exp(-r\epsilon^2/8) \leq 4 \exp(-n/\epsilon)\,.
$$
\end{proof}

An immediate corollary is that, for $\mathcal R$ as in the theorem, $\mathcal I (\mathcal R)$  is in PromiseBPP. We will also need the following simple result about simulating circuits constructed from conjugated Clifford gates.

\begin{theorem}\label{thm:clifford}
Let $S \in \mathcal C_2$, and $R = (Q \otimes Q) S (Q \otimes Q)^\dagger$ where $Q$ is a single-qubit gate. Let $U$ be a $\{R\}$-circuit on $n$ qubits, $M$ a Hermitian operator on $O(\log(n))$ qubits, and $\ket{\psi}, \ket{\phi}$ arbitrary $n$-qubit product states. Then $\bra{\psi} U^\dagger (M \otimes I) U \ket{\phi}$ can be computed exactly in $O(\text{poly}(n))$ classical time.
\end{theorem}
\begin{proof}
We first apply the procedure from \eqref{eq:pull-q-out} as before, and write
$$
U = Q^{\otimes n} V (Q^\dagger)^{\otimes n}
$$
where $V$ is described by a circuit consisting only of $S$ gates. The unitary operator implemented by $V$ is an element of $\mathcal C_n$. Now let $M$ be a Hermitian operator on $m = c \log (n)$ qubits, and suppose for simplicity that it acts only on the first $m$ qubits. Let $I$ denote the identity operator on the $(m+1)$st through $n$th qubits. We write
\begin{align*}
\bra{\psi} U^\dagger (M \otimes I) U \ket{\phi} 
&= \bra{\psi} Q^{\otimes n} V^\dagger Q^{\dagger \otimes n} 
(M \otimes I) Q^{\otimes n} V Q^{\dagger \otimes n} \ket{\phi} \\
&= \bra{\psi} Q^{\otimes n} V^\dagger (M' \otimes I) V Q^{\dagger \otimes n} \ket{\phi}\,,
\end{align*}
where $M' = Q^{\otimes m} M Q^{\dagger \otimes m}$. 

As discussed earlier, a basis for the space of Hermitian operators on $m$ qubits is the $m$-qubit Pauli group $\mathcal P_m$, which has size $O(\text{poly}(n))$. The expansion of $M'$ in that basis can be computed in polynomial time by basic linear algebra. Embedding the first $m$ qubits into all $n$ qubits gives the obvious embedding of $\mathcal P_m$ into $\mathcal P_n$, and this also gives (the same, still polynomial-size) expansion of $M'$ into $n$-qubit Paulis. We write
$$
M' = \sum_{\sigma \in \mathcal P_n \cap \mathcal P_m} \alpha_\sigma \sigma\,.
$$
We emphasize that this is a sum over polynomially many terms, and that each coefficient can be calculated from knowledge of $M$ and $Q$ in polynomial time. Moreover, since $V$ is a Clifford circuit, its conjugation action $\sigma \mapsto \sigma^V := V^\dagger \sigma V$ on a Pauli group element $\sigma \in \mathcal P_n$ is easily computed by direct gate-by-gate matrix multiplication (see, e.g.,~\cite{Gottesman}). 

We now return to the main calculation, to see that
\begin{align*}
\bra{\psi} U^\dagger (M \otimes I) U \ket{\phi} 
&= \bra{\psi} Q^{\otimes n} V^\dagger (M' \otimes I) V Q^{\dagger \otimes n} \ket{\phi}\\
&= \sum_{\sigma \in \mathcal P_n \cap \mathcal P_m} \alpha_\sigma \bra{\psi} Q^{\otimes n} V^\dagger \sigma V Q^{\dagger \otimes n} \ket{\phi}\\
&= \sum_{\sigma \in \mathcal P_n \cap \mathcal P_m} \alpha_\sigma \bra{\psi} Q^{\otimes n}\sigma^V Q^{\dagger \otimes n} \ket{\phi}\\
&= \sum_{\sigma \in \mathcal P_n \cap \mathcal P_m} \alpha_\sigma \prod_{j=1}^n \bra{\psi_j} Q \sigma^V_j Q^\dagger \ket{\phi_j}\,.
\end{align*}
The sum and product in the final expression are both of polynomial size, and each term in the product can be computed in constant time.
\end{proof}

\section{Qubit solutions to Yang-Baxter}

\subsection{The four solution families}

Hietarinta classified all solutions to the Yang-Baxter equation in the qubit (i.e., $4 \times 4$) case~\cite{Hietarinta93}. The qubit solutions which are also unitary operators were identified by Dye~\cite{Dye03}. All of these are of the form 
\begin{equation}\label{eq:ybe-solutions}
R = k(Q \otimes Q) ST (Q \otimes Q)^{-1}
\end{equation}
where $k$ is a unit-norm scalar, $T$ is the swap gate, and
$$
Q = 
\begin{pmatrix} 
	a & b \\ 
	c & d 
\end{pmatrix}
$$
is an invertible matrix. The trivial solution is $S = T$ which implies $R = k \identity$. There are four nontrivial solution families, depending on the possible values taken by $S$, which are listed below, along with the required conditions on the matrix entries.
\begin{align*}
S_1 &=
\begin{pmatrix} 1 & 0 & 0 & 0\\
      0 & p & 0 & 0\\
      0 & 0 & q & 0\\
      0 & 0 & 0 & r\\
    \end{pmatrix} 
~~~~ &1 = |p| = |q| = |r| \text{ ; } c = -a\bar{b}/\bar{d}\\
S_2 &= \begin{pmatrix} 0 & 0 & 0 & p\\
      0 & 0 & 1 & 0\\
      0 & 1 & 0 & 0\\
      q & 0 & 0 & 0\\
    \end{pmatrix}
~~~~&p = \frac{(b \bar b + d \bar d)(\bar a b + \bar c d)}{(a \bar a + c \bar c)(a \bar b + c \bar d)} \text{ ; } q = 1/p \text{ ; } c \neq -a\bar b / \bar d\\
S_3 &= \begin{pmatrix} 0 & 0 & 0 & p\\
      0 & 0 & 1 & 0\\
      0 & 1 & 0 & 0\\
      q & 0 & 0 & 0\\
    \end{pmatrix}
~~~~&p\bar p = \frac{(d \bar d)^2}{(a \bar a)^2} \text{ ; } q \bar q = \frac{(a \bar a)^2}{(d \bar d)^2} \text{ ; } |pq| = 1 \text{ ; } c = -a\bar{b}/\bar{d}\\
S_4 &= 
\frac{1}{\sqrt{2}}\begin{pmatrix} 1 & 0 & 0 & 1 \\
      0 & 1  & 1  & 0\\
      0 & 1  & -1  & 0\\
      -1  & 0 & 0 & 1 \\
    \end{pmatrix}
~~~~&|a| = |d| \text{ ; } c = -a \bar b / \bar d\,.
\end{align*}

For $j = 1, 2, 3, 4$, let $R_j$ be the Yang-baxter operator~\eqref{eq:ybe-solutions} resulting from choosing $S = S_j$. 

\subsection{Families one, two and three are unlikely to be universal}

We will show that Theorem \ref{thm:main} applies to the single-element gate sets $\{R_1\}$, $\{R_2\}$, and $\{R_3\}$. We assume that all of the above matrix entries are exactly computable in constant time via a classical algorithm.

The gate $R_1$ has the form \eqref{eq:general-gate-form} where $C_i = \identity$, $P_i = T$, and $D_i = kS_1$. It remains to check that $Q$ satisfies property $(G)$ where $G$ is the trivial group consisting only of the identity; this is confirmed by Lemma \ref{lem:property-1} below. 

For the gate $R_2$, we set
$$
M = 
\begin{pmatrix} 
	0 & \sqrt{p} \\ 
	1/\sqrt{p} & 0 
\end{pmatrix}
$$
and check that $M \otimes M = S_2$. It follows that $R_2 = kT (QMQ^{-1} \otimes QMQ^{-1})$ is not an entangling gate. Since $R_2$ is unitary, so is $QMQ^{-1}$. By the spectral theorem, there exist diagonal $V$ and unitary $U$ such that $UVU^{-1} = QMQ^{-1}$. Observe that $R_2 = (U \otimes U) k(V \otimes V)T (U \otimes U)^{-1}$ satisfies the conditions of Theorem \ref{thm:main}. 

For the gate $R_3$, we first rewrite the matrices as follows. Set
$$
N = 
\begin{pmatrix} 
	p^{-1/4} & 0 \\ 
	0 & p^{1/4} 
\end{pmatrix}
$$
and $Q' = QN^{-1}$ and $S_3' = (N \otimes N) S_3 (N \otimes N)^{-1}$. It's not hard to check that
$$
R_3 = k (Q \otimes Q) S_3 T (Q \otimes Q)^{-1} = k(Q' \otimes Q') S_3' T (Q' \otimes Q')^{-1}~,
$$
and that $Q'$ and $S_3'$ satisfy the conditions of the third YBE solution family, with the additional property that $p = 1$ and $|q| = 1$. Note further that $S_3' (X \otimes X)$ is a diagonal unitary operator, where
$$
X = 
\begin{pmatrix} 
	0 & 1 \\ 
	1 & 0 
\end{pmatrix}\,.
$$
We now see that $R_3$ is of the form \eqref{eq:general-gate-form} from Theorem \ref{thm:main}, where $C_i = X$, $D_i = kS_3' (X \otimes X)$, and $P = T$. It remains to check that $Q'$ satisfies property $(\langle X \rangle)$, which is done in Lemma \ref{lem:property-2} below.

\begin{lemma}\label{lem:property-1}
Let $Q$ be an invertible $2 \times 2$ matrix defined by
$$
Q = \begin{pmatrix} 
	a & b \\ 
	c & d 
\end{pmatrix}\,,
$$
such that $c = -a \bar b / \bar d$. Then $Q$ satisfies property $(\identity)$.
\end{lemma}
\begin{proof}
Define the relevant matrices
$$
A = \begin{pmatrix} 
	|a| & |b| \\ 
	|c| & |d| 
\end{pmatrix}
\qquad \text{and} \qquad
B = \frac{1}{|ad-bc|}
\begin{pmatrix} 
	|d| & |b| \\ 
	|c| & |a| 
\end{pmatrix}\,.
$$
Note that $a = 0$ implies $c = 0$, which would make $Q$ non-invertible. 

We compute each case separately. First let $k = l = 1$.
\begin{align*}
A_{11}B_{11} + A_{12}B_{21}
& = \frac{|a||d| + |b||c|}{|ad-bc|} 
	= \frac{|a||d| + |b|| a \bar b / \bar d|}{|ad + b a \bar b / \bar d|}\\
& = \frac{|\bar d|(|a||d|^2 + |a||b|^2)}{|\bar d|(|a d \bar d + a b \bar b|)}
 	= \frac{|a|(|d|^2 + |b|^2)}{|a||d \bar d + b \bar b|} = 1\,.
\end{align*}
Next, let $k = l = 2$, and we again get
$$
A_{21}B_{12} + A_{22}B_{22} = \frac{|c||b| + |d||a|}{|ad-bc|} = 1\,.
$$
Now suppose $k=1$ and $l=2$. 
\begin{align*}
A_{11}B_{12} + A_{12}B_{22}
& = \frac{|a||b| + |b||a|}{|ad-bc|} 
	= \frac{2|a||b|}{|ad+ab\bar b/\bar d|} \\
& = \frac{2|a||b||d|}{|ad\bar d + ab\bar b|} 
	= \frac{2|b||d|}{|d|^2 + |b|^2}\,. 
\end{align*}
It remains to note that
$$
|b|^2 + |d|^2 - 2|b||d| = (|b| - |d|)^2 > 0\,.
$$
Finally, we choose $k=2$ and $l=1$.
$$
A_{21}B_{11} + A_{22}B_{21}
= \frac{|c||d| + |d||c|}{|ad-bc|} 
= \frac{2|a||b|}{|ad-bc|}  \leq 1\,,
$$
by two applications of $c = -a \bar b / \bar d$ and the previous case.
\end{proof}

\begin{lemma}\label{lem:property-2}
Let $Q$ be an invertible $2 \times 2$ matrix defined by
$$
Q = \begin{pmatrix} 
	a & b \\ 
	c & d 
\end{pmatrix}\,,
$$
such that $c = -a \bar b / \bar d$ and $|a|^2 = |d|^2$. Then $Q$ satisfies property $(\mathcal S_2)$.
\end{lemma}
\begin{proof}
Define the matrices $A$ and $B$ as in Lemma \ref{lem:property-1}. The case of $\pi$ equal to the trivial permutation is handled by Lemma \ref{lem:property-1}. We compute the remaining cases. Set $\pi = (12)$ and $k = l = 1$. Then
\begin{align*}
A_{12}B_{11} + A_{11}B_{21} 
& = \frac{|a||c| + |b||d|}{|ad-bc|}
	= \frac{|aa \bar b / \bar d| + |bd|}{|ad-a b \bar b/\bar d|} \\
& = \frac{|aa\bar b| + |b d \bar d|}{|ad \bar d + a b \bar b|}
= \frac{|aa \bar b| + |b a \bar a|}{|a a \bar a + a b \bar b|}
= \frac{|a \bar b| + |b \bar a|}{|a|^2 + |b|^2}\,,
\end{align*}
where we have applied the facts $c = -a \bar b / \bar d$ and $a \bar a = d \bar d$ and $a \neq 0$. Now note that
$$
|a|^2 + |b|^2 - (|a \bar b| + |b \bar a|) = |a|^2 + |b|^2 - 2|a||b| = (|a|-|b|)^2 \geq 0\,.
$$
Hence $(|a \bar{b}| + |b \bar{a}|) / (|a|^2 + |b|^2) \leq 1$. For $k = l = 2$, we again get
$$
A_{22}B_{12} + A_{21}B_{22} = \frac{|a||c| + |b||d|}{|ad-bc|} \leq 1\,.
$$
Now set $k = 1$ and $l = 2$. Then
\begin{align*}
A_{12}B_{12} + A_{11}B_{22} 
&= \frac{|a|^2 + |b|^2}{|ad-bc|} 
	= \frac{|a|^2 + |b|^2}{|ad + ab\bar b/\bar d|} \\
&= \frac{|\bar d|(|a|^2 + |b|^2)}{|ad\bar d + ab\bar b|}
	= \frac{|d|(|a|^2 + |b|^2)}{|a|(|d|^2 + |b|^2)} = 1\,.
\end{align*}
Finally, for $k=2$ and $l=1$, write $b = -\bar c d / \bar a$ and calculate
\begin{align*}
A_{22}B_{11} + A_{21}B_{21} 
&= \frac{|c|^2 + |d|^2}{|ad-bc|} 
	= \frac{|c|^2 + |d|^2}{|ad + dc\bar c/\bar a|} \\
&= \frac{|\bar a|(|c|^2 + |d|^2)}{|d a\bar a + dc\bar c|}
	= \frac{|a|(|c|^2 + |d|^2)}{|d|(|c|^2 + |a|^2)} = 1\,.
\end{align*}
\end{proof}

To conclude, we have shown the following.

\begin{theorem}
Let $R \in \{R_1, R_2, R_3\}$ be a unitary solution to the Yang-Baxter equation on qubits. Then $\mathcal I(\{R\})$ is in PromiseBPP.
\end{theorem}

\noindent In particular, if one could perform (perhaps encoded) universal quantum computation with these circuits then PromiseBQP $=$ PromiseBPP. We can also formulate the lack of universality for these solutions in the following terms.

\begin{theorem}
Let $R \in \{R_1, R_2, R_3\}$ be a unitary solution to the Yang-Baxter equation on qubits, and let $\rho_n: B_n \rightarrow SU(2^n)$ be the corresponding unitary representation of the braid group. Then the image of $\rho_n$ is not dense in $SU(2^n)$ for any $n \geq 2$, unless \emph{PromiseBQP} $=$ \emph{PromiseBPP}.
\end{theorem}
\begin{proof}
(Sketch) For a contradiction, suppose there exists an $n \geq 2$ such that the image of $\rho_n$ is dense. Let $C$ be an arbitrary $m$-qubit quantum circuit. We can assume without loss of generality that $C$ only consists of $2$-qubit gates acting on adjacent qubits, and that $n$ is even. For each of the $m$ qubits, assign $n/2$ qubits from the space of $\rho_n$. By the density of the image of $\rho_n$, we can then simulate $C$ inside $\rho_{mn/2}$ gate-by-gate via the Solovay-Kitaev theorem. Then we can use the classical algorithm from Theorem \ref{thm:main} to approximate the relevant matrix entry of the resulting $R$-circuit, thus solving the PromiseBQP-hard problem of approximating the corresponding entry of $C$.
\end{proof}

\subsection{Family four is unlikely to be universal}

Recall that the fourth solution family is of the form $R_4 = k (Q \otimes Q)S_4T(Q\otimes Q)^{-1}$. We begin by demonstrating a Clifford circuit which is equal to the gate $S_4T $.

$$
\raisebox{-11pt}{
$S_4 T
	= \frac{1}{\sqrt{2}}
	\begin{pmatrix} 1 & 0 & 0 & 1 \\
			      	0 & 1  & 1  & 0\\
				0 & -1  & 1  & 0\\
			       -1  & 0 & 0 & 1 \\
	\end{pmatrix} =~~
$}  
\Qcircuit @C=1em @R=.5em {
&\gate{Z} &\ctrl{1} &\gate{X} &\gate{H} &\qw\\
&\gate{Z} &\targ &\gate{Z} &\qw &\qw}
$$

We also note that, in this solution family, $Q$ is a scaled unitary operator. To see this, note that
$$
Q^\dagger Q = 
\begin{pmatrix} 
	|a|^2 + |c|^2 & \bar a b + \bar c d \\ 
	a \bar b + c \bar d & |b|^2 + |d|^2
\end{pmatrix}
= \begin{pmatrix} 
	|a|^2 + |c|^2 & 0 \\ 
	0 & |b|^2 + |a|^2
\end{pmatrix}
= \bigl(|a|^2 + |b|^2\bigr)
\begin{pmatrix} 
	1 & 0 \\ 
	0 & 1
\end{pmatrix}
$$
where we first applied the condition $c = - a \bar b / \bar d$ to the off-diagonal elements and the condition $|a|^2 = |d|^2$ to the diagonal ones; the last equality follows from combining these two conditions to get $|c|^2 = |b|^2$. Now set $\alpha = (|a|^2 + |b|^2)^{1/2}$ and $Q_1 = \alpha^{-1} Q$. Using the above, one easily checks that $Q_1$ is unitary and that $Q_1^\dagger = \alpha Q^{-1}$. It follows that
$$
(Q \otimes Q) A (Q \otimes Q)^{-1} 
= (\alpha Q_1 \otimes \alpha Q_1) A (\alpha^{-1}Q_1^\dagger \otimes \alpha^{-1} Q_1^\dagger) 
= (Q_1 \otimes Q_1) A (Q_1 \otimes Q_1)^\dagger
$$
for any $A$. For us it will thus suffice to assume that $Q$ is in fact unitary. This allows us to apply Theorem \ref{thm:clifford} and get the following result.

\begin{theorem}
Let $U$ be a $\{R_4\}$-circuit on $n$ qubits, $M$ a Hermitian operator on $O(\log(n))$ qubits, and $\ket{\psi}, \ket{\phi}$ arbitrary $n$-qubit product states. Then $\bra{\psi} U^\dagger (M \otimes I) U \ket{\phi}$ can be computed exactly in $O(\text{poly}(n))$ classical time.
\end{theorem}

\section{Some simple high-dimensional solutions}

Finally, we list some simple unitary solution families to the Yang-Baxter equation that exist in every dimension, and to which Theorem \ref{thm:main} applies. We begin by observing that, whenever a $2$-qudit gate $S$ is a solution, then by \eqref{eq:pull-q-out} so is $(Q \otimes Q) S (Q \otimes Q)^{-1}$ for any $1$-qudit gate $Q$.

For $A, B \in \opnm U(V)$, the operator $T(A \otimes B)$ is a solution to the Yang-Baxter equation if and only if $A$ and $B$ commute. This is easily seen by following the wires in the circuits below.
$$
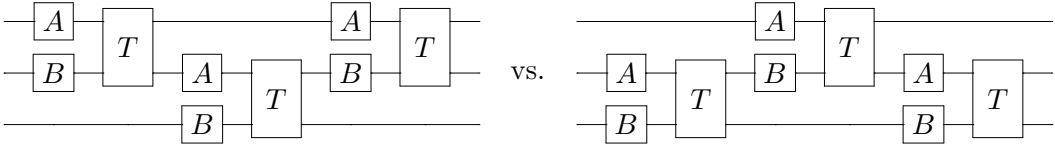
\begin{tabular}{ l c r }
\Qcircuit @C=1em @R=.5em {
&\gate{A}	&\multigate{1}{T}	&\qw		&\qw				&\gate{A} 	&\multigate{1}{T}	&\qw\\
&\gate{B} 	&\ghost{T}		&\gate{A}	&\multigate{1}{T}	&\gate{B} 	&\ghost{T}		&\qw\\
&\qw		&\qw				&\gate{B}	&\ghost{T}		&\qw 	&\qw 			&\qw
}
& \raisebox{-22pt}{\text{vs.}} & 
\Qcircuit @C=1em @R=.5em {
&\qw		&\qw				&\gate{A}	&\multigate{1}{T}	&\qw 	&\qw 			&\qw\\
&\gate{A}	&\multigate{1}{T}	&\gate{B} 	&\ghost{T}		&\gate{A} 	&\multigate{1}{T}	&\qw\\
&\gate{B}	&\ghost{T}		&\qw		&\qw				&\gate{B} 	&\ghost{T}		&\qw
}
\end{tabular}
$$
If $A$ and $B$ do commute, then there's a unitary change of basis $Q$ on $V$ such that $Q^{-1}AQ$ and $Q^{-1}BQ$ are both diagonal. Therefore, Theorem~\ref{thm:main} applies to $T(A \otimes B)$, so any circuits using this gate are classically simulable. Of course, this is not surprising, as they do not even entangle the qudits.

More generally, suppose $S \in \opnm U(V \otimes V)$ is diagonal in the computational basis, and set
$$
\lambda_{ij} = \bra{ij} S \ket{ij} \qquad \text{for} \qquad i, j \in [d]\,,
$$
where $d = \dim V$. Note that
$$
S_{12} = S \otimes \identity = \bigoplus_{k \in [d]} P_k\,,
\qquad
S_{23} = \identity \otimes S = \bigoplus_{k \in [d]} S\,,
\qquad
\identity \otimes T = \bigoplus_{k \in [d]} T\,.
$$
where $P_k = \oplus_{l \in [d]} \lambda_{kl} \identity$. We also have
$$
S_{13} 
= (\identity \otimes T) (S \otimes \identity) (\identity \otimes T)
= \bigoplus_{k \in [d]} TP_kT\,.
$$
Substituting the above into the two sides of the algebraic Yang-Baxter equation \eqref{eq:AYBE}, we get
$$
\bigoplus_{k \in [d]} P_kTP_kTS
\qquad \text{and} \qquad
\bigoplus_{k \in [d]} STP_kTP_k
$$
Clearly, $P_k$ and $S$ are symmetric. Since $\bra{ab} T \ket{cd} = \delta_{ad}\delta_{bc} = \bra{cd} T \ket{ab}$, so is $T$. By applying the transpose to one of the two sides above, we see that $S$ satisfies algebraic Yang-Baxter. Thus $ST$ is a solution to the YBE, one to which Theorem \ref{thm:main} clearly applies.
\section{Appendix}
\label{sec:app}
We will now prove Lemmas \ref{lem:sampling} and \ref{lem:chernoff}.
\addtocounter{lemma}{-4}
\begin{lemma}\label{lem:sampling}
Let $\{P_j\}_{j=1}^n$ be probability distributions on $[d]$ and let $P = \Pi_j P_j$ be the corresponding product distribution over $[d]^n$. Suppose that we can calculate $P_j(k)$ for every $j$ and every $k$ in total time $\emph{poly(n, d)}$. Then there's a classical probabilistic algorithm that runs in time $\emph{poly(n, d)}$ and samples from $[d]^n$ according to a probability distribution $D$ such that $|P - D| \leq 1/2^{\emph{poly}(n)}$.
\end{lemma}
\begin{proof}
To sample from $P_j$, flip $m$ unbiased coins to get an integer $0 \leq l \leq 2^{m}$. Subdivide $2^{m}$ into intervals according to 
$$
2^{m} = P_j(0)2^m + P_j(1)2^m + \cdots + P_j(d-1)2^m
$$
and output $k$ if $l$ falls into the $k$th interval. Then the probability $D_j(k)$ with which you output $k$ satisfies $|D_j(k) - P_j(k)| \leq 1/2^m$. Now do this for two indices, say $1$ and $2$ and note that
\begin{align*}
|P_1(k)P_2(l) - D_1(k)D_2(l)| 
&= |P_1(k)P_2(l) - D_1(k)D_2(l) + D_1(k)P_2(l) - D_1(k)P_2(l)|\\
&\leq  |P_2(l)(P_1(k) - D_1(k))| + |D_1(k)(P_2(l) - D_2(l))|\\
&\leq 2/2^m
\end{align*}
Extending this to the case of multiplying all $n$ distributions together, we get $|P(y)-D(y)| \leq n/2^m$ for all $y \in [d]^n$. The total variation distance then satisfies
$$
|P - D| = \frac{1}{2}\sum_{x \in [d]^n} |P(x) - D(x)| \leq \frac{n d^n}{2^m} < 2^{-n}
$$
so long as $m \geq 3n \log d$.
\end{proof}

\begin{lemma}\label{lem:chernoff}
Let $X_1, X_2, \ldots, X_n$ be independent complex-valued random variables with $\EE[X_j] = \mu$ and $|X_j| \leq b$ for all $j$. Let $S = \sum_j X_j/n$. Then
$$
\emph{Pr}\left[ \left|S - \mu \right| \geq \epsilon \right] \leq 4 \exp\left(-n\epsilon^2/8b^2\right)\,.
$$
\end{lemma}
\begin{proof}
We expand the $X_j$ into real and imaginary parts and apply the standard bound. Set $S_r = \sum_j \re[X_j]/n$ and $S_i = \sum_j \im[X_j]/n$ and $\mu_r = \EE[\re[X_j]]$ and $\mu_i = \EE[\im[X_j]]$. Note that $|\re[X_j]| \leq b$ and $|\im[X_j]| \leq b$. By the Chernoff-Hoeffding bound for real-valued random variables~\cite{Hoeffding63}, we have
$$
\text{Pr}\left[ \left|S_r - \mu_r \right| \geq \epsilon/2 \right] \leq 2 \exp\left(-n\epsilon^2/8b^2\right)\,,
$$
and likewise for the imaginary part. Taking the union bound, we have that 
$$
| S - \mu | = \left| S_r - \mu_r + i(S_i - \mu_i) \right| \leq |S_r - \mu_r| + |S_i - \mu_i| \leq \epsilon/2 + \epsilon/2 = \epsilon
$$
except with probability $4 \exp\left(-n\epsilon^2/8b^2\right)$.
\end{proof}

\section{Acknowledgments}
We thank M{\=a}ris Ozols for pointing out an error in a previous version of this manuscript. We acknowledge funding provided by the Institute for Quantum Information and Matter, an NSF Physics Frontiers Center with support of the Gordon and Betty Moore Foundation through Grant GBMF1250. We would also like to acknowledge the support from the Summer Undergraduate Research Fellowship (SURF) program, as well as the David L. Goodstein SURF endowment. Portions of this paper are a contribution of NIST, an agency of the US government, and are not subject to US copyright.

\bibliography{ybe}

\end{document}